\documentclass[journal,twoside]{IEEEtran}
\usepackage{cite}
\usepackage[pdftex]{graphicx}
\usepackage[cmex10]{amsmath}
\usepackage{algorithmic,algorithm}
\usepackage{array}
\usepackage{mdwmath}
\usepackage{mdwtab}
\usepackage{eqparbox}
\usepackage[tight,footnotesize]{subfigure}
\usepackage{epsfig}
\usepackage{url}
\usepackage{fancyhdr}
\usepackage{amsmath,amsfonts,amssymb,mathrsfs}
\usepackage[pdftex,bookmarksnumbered,bookmarksopen]{hyperref}
\usepackage{makecell}

\newtheorem{theorem}{Theorem}[section]

\newenvironment{proof}[1][Proof:]{\begin{trivlist}
\item[\hskip \labelsep {\bfseries #1}]}{\end{trivlist}}

\makeatletter

\newcommand{\Rmnum}[1]{\expandafter\@slowromancap\romannumeral #1@}
\makeatother

\long\def\symbolfootnote[#1]#2{\begingroup\def\thefootnote{\fnsymbol{footnote}}\footnote[#1]{#2}\endgroup}

\begin{document}

\title{On the Confidentiality of Information Dispersal Algorithms and Their Erasure Codes}
\author{Mingqiang Li\\
Department of Computer Science and Engineering, The Chinese University of Hong Kong\\
Shatin, New Territories, Hong Kong\\
\mbox{Email}: mingqiangli.cn@gmail.com\\
\thanks{The main part of this work was finished while Mingqiang Li worked as a Staff Researcher in the IBM China Research Laboratory.}
}

\maketitle

\begin{abstract}
\emph{Information Dispersal Algorithms (IDAs)} have been widely applied to reliable and secure storage and transmission of data files in distributed systems. An IDA is a method that encodes a file $F$ of size $L=|F|$ into $n$ unrecognizable pieces $F_1$, $F_2$, $\cdots$, $F_n$, each of size $L/m$ ($m<n$), so that the original file $F$ can be reconstructed from any $m$ pieces. The core of an IDA is the adopted \mbox{non-systematic} \mbox{$m$-of-$n$} erasure code. This paper makes a systematic study on the \emph{confidentiality} of an IDA and its connection with the adopted erasure code. Two levels of confidentiality are defined: \emph{weak confidentiality} (in the case where some parts of the original file $F$ can be reconstructed explicitly from fewer than $m$ pieces) and \emph{strong confidentiality} (in the case where nothing of the original file $F$ can be reconstructed explicitly from fewer than $m$ pieces). For an IDA that adopts an arbitrary \mbox{non-systematic} erasure code, its confidentiality may fall into weak confidentiality. To achieve strong confidentiality, this paper explores a sufficient and feasible condition on the adopted erasure code. Then, this paper shows that Rabin's IDA has strong confidentiality. At the same time, this paper presents an effective way to construct an IDA with strong confidentiality from an arbitrary \mbox{$m$-of-$(m+n)$} erasure code. Then, as an example, this paper constructs an IDA with strong confidentiality from a Reed-Solomon code, the computation complexity of which is comparable to or sometimes even lower than that of Rabin's IDA.
\end{abstract}

\begin{IEEEkeywords}
Cauchy matrix, confidentiality, erasure code, information dispersal algorithm, \mbox{Reed-Solomon} code, Vandermonde matrix.
\end{IEEEkeywords}

\section{Introduction}\label{section-intro}
\indent In 1989, Rabin \cite{IDA} proposed an attractive \emph{Information Dispersal Algorithm (IDA)} that is applicable to reliable and secure storage and transmission of data files in distributed systems. Since then, IDAs have drawn many attentions from both researchers and engineers in the area of distributed systems.\\
\indent  An IDA is a method that encodes a file $F$ of size $L=|F|$ into $n$ unrecognizable pieces $F_1$, $F_2$, $\cdots$, $F_n$, each of size $L/m$ ($m<n$), so that the original file $F$ can be reconstructed from any $m$ pieces. From a coding theorist's viewpoint, an IDA is corresponding to a \mbox{non-systematic} \mbox{$m$-of-$n$} erasure code \cite{EEC}. Here, the \mbox{non-systematic} property of the erasure code is necessary to ensure ``unrecognizable'' pieces. In practice, an IDA is implemented as follows: The original file $F$ is firstly divided into $m$ segments $S_1$, $S_2$, $\cdots$, $S_m$, each of size $L/m$. Then, the $m$ segments are encoded into $n$ unrecognizable pieces $F_1$, $F_2$, $\cdots$, $F_n$ using a \mbox{non-systematic} \mbox{$m$-of-$n$} erasure code.\\
\indent The reliability of an IDA is clear: no more than $n-m$ lost pieces of the $n$ pieces $F_1$, $F_2$, $\cdots$, $F_n$ will not result in data loss. However, the \emph{confidentiality} of an IDA is not straightforward and deserves a systematic study.\\
\indent From the view of \mbox{information-theoretic} security \cite{Shannon-perfect-secrecy}, IDAs can provide only incremental confidentiality and thus have weaker confidentiality than secret sharing (with perfect confidentiality) \cite{secret-sharing-Blakley,secret-sharing-Shamir,secret-sharing-condition} and ramp schemes (with partially perfect confidentiality) \cite{ramp-scheme-1,ramp-scheme-2}. However, IDAs can achieve optimal efficiency in data overhead \cite{IDA}. As shown in \cite[Page 65, Table A]{threshold-compare}, there is a \mbox{trade-off} between confidentiality and data overhead. Moreover, for practical applications, \mbox{information-theoretic} security is often extravagant and unnecessary. Thus, in this paper, we will study the practical security that IDAs can provide.\\
\indent Although a \mbox{non-systematic} erasure code can ensure ``unrecognizable'' pieces in an IDA, some segments may still be reconstructed explicitly from fewer than $m$ pieces. Then, an eavesdropper who acquires fewer than $m$ pieces by snooping may reconstruct some parts of the original file $F$ explicitly, resulting in partial file leakage. In the case of partial file leakage, we say the IDA has \emph{weak confidentiality}. However, for an ideal IDA, any segment of the original file $F$ should not be reconstructed explicitly from fewer than $m$ pieces. In the case where nothing of the original file $F$ can be reconstructed explicitly from fewer than $m$ pieces, we say the IDA has \emph{strong confidentiality}.\\
\indent For an IDA that adopts an arbitrary \mbox{non-systematic} erasure code, we noticed that its confidentiality may fall into weak confidentiality. In this paper, we will first show in Section~\ref{section-weak-conf} which kind of IDAs has weak confidentiality and how an eavesdropper can reconstruct some segments of the original file $F$ explicitly from fewer than $m$ pieces in the case of weak confidentiality. Then, to achieve strong confidentiality, we explore a sufficient and feasible condition for an IDA in Section~\ref{section-strong-conf}. We show that Rabin's IDA \cite{IDA} has strong confidentiality. At the same time, we present an effective way to construct an IDA with strong confidentiality from an arbitrary \mbox{$m$-of-$(m+n)$} erasure code. Then, as an example, we construct an IDA with strong confidentiality from a Reed-Solomon code \cite{RS}, the computation complexity of which is comparable to or sometimes even lower than that of Rabin's IDA. Finally, we conclude this paper in Section~\ref{section-conclusion}. To our knowledge, this paper is the first work that focuses on the issues of weak confidentiality and strong confidentiality in IDAs.\\
\indent To make our later discussion more easily understood, we begin this paper with a brief introduction of IDAs and their erasure codes.

\section{IDAs and Their Erasure Codes}
\indent In an Information Dispersal Algorithm (IDA), a \mbox{non-systematic} \mbox{$m$-of-$n$} erasure code is employed to encode the $m$ segments $S_1$, $S_2$, $\cdots$, $S_m$ into $n$ unrecognizable pieces $F_1$, $F_2$, $\cdots$, $F_n$, i.e.
\begin{equation}\label{equ-encode}
    (S_1,S_2,\cdots,S_m)\cdot G_{m\times n}=(F_1,F_2,\cdots,F_n),
\end{equation}
where $G_{m\times n}$ is the \emph{generator matrix} of the adopted erasure code and meets the following two conditions:
\begin{enumerate}
  \item Any column of $G_{m\times n}$ is not equal to any column of an $m\times m$ identity matrix; and
  \item Any $m$ columns of $G_{m\times n}$ form an $m\times m$ nonsingular matrix.
\end{enumerate}
The first condition ensures that any piece is unrecognizable; while the second condition ensures that the original file $F$ can be reconstructed from any $m$ pieces.

\section{Which Kind of IDAs Has Weak Confidentiality}\label{section-weak-conf}
\indent In this section, we will show which kind of Information Dispersal Algorithms (IDAs) have weak confidentiality and how an eavesdropper can reconstruct some segments of the original file $F$ explicitly from fewer than $m$ pieces in the case of weak confidentiality. We present a theorem as follows:
\begin{theorem}\label{th-weak}
An IDA has \emph{weak confidentiality} if and only if the adopted erasure code meets the following condition: In its generator matrix $G_{m\times n}$, there is a submatrix $A_{m'\times n'}$ of column rank $r$, where $m',n'<m$ and $n'-r=m-m'>0$.
\end{theorem}
\begin{proof}
We first prove the sufficiency. Suppose $A_{m'\times n'}$ is located in rows $i_1,i_2,\cdots,i_{m'}$ and columns $j_1,j_2,\cdots,j_{n'}$ of $G_{m\times n}$. Then, $S_{i_1},S_{i_2},\cdots,S_{i_{m'}}$ are the $m'$ segments corresponding to rows $i_1,i_2,\cdots,i_{m'}$ of $G_{m\times n}$. Similarly, $F_{j_1},F_{j_2},\cdots,F_{j_{n'}}$ are the $n'$ pieces corresponding to columns $j_1,j_2,\cdots,j_{n'}$ of $G_{m\times n}$. Since $n'-r=m-m'>0$, $r=m'+n'-m$. Furthermore, since $m',n'<m$, then $r<m',n'$. Thus, $A_{m'\times n'}$ is rank deficient. Then, in $A_{m'\times n'}$, any $k=n'-r$ columns can be linearly represented by other $r$ columns. Let $A_{m'\times n'}=(v_1,v_2,\cdots,v_{n'})$, where $v_1,v_2,\cdots,v_{n'}$ are column vectors. Suppose there is a linear relation among column vectors of $A_{m'\times n'}$ as follows: \[(v_1,v_2,\cdots,v_{k})=(v_{k+1},v_{k+2},\cdots,v_{n'})\cdot B_{r\times k},\] where $B_{r\times k}$ is the transpose of coefficient matrix. Then, any information of $S_{i_1},S_{i_2},\cdots,S_{i_{m'}}$ can be eliminated by calculating
\begin{multline}\label{equ-recon}
    (\tilde{F}_{j_1},\tilde{F}_{j_2},\cdots,\tilde{F}_{j_{k}})=(F_{j_1},F_{j_2},\cdots,F_{j_{k}})\\
    -(F_{j_{k+1}},F_{j_{k+2}},\cdots,F_{j_{n'}})\cdot B_{r\times k}.
\end{multline}
Finally, according to Equation~(\ref{equ-encode}), other $m-m'=k$ segments except $S_{i_1},S_{i_2},\cdots,S_{i_{m'}}$ can be reconstructed explicitly.  \\
\indent We now prove the necessity. If some segments can be reconstructed explicitly from fewer than $m$ pieces in an IDA, it is clear that the information of other segments should be able to be eliminated from the eavesdropped pieces by linear operations. Moreover, a solvable system of linear equations on these reconstructible segments should be able to be formed. From the above proof of sufficiency, we can deduce that in the corresponding generator matrix $G_{m\times n}$, there should be a submatrix $A_{m'\times n'}$ of column rank $r$, where $m',n'<m$ and $n'-r=m-m'>0$.
\end{proof}

\indent \emph{Remark:} From the second condition in the previous section, we can deduce a necessary condition $n'-r\leq m-m'$ (otherwise the corresponding $n'$ columns of $G_{m\times n}$ will form a rank deficient matrix---a contradiction). Thus, for an IDA that adopts an arbitrary \mbox{non-systematic} erasure code, its confidentiality may fall into weak confidentiality.

\section{Constructing IDAs with Strong Confidentiality}\label{section-strong-conf}
\indent For an Information Dispersal Algorithm (IDA), to achieve strong confidentiality, we explore a sufficient and feasible condition as follows:
\begin{theorem}\label{th-strong}
An IDA has strong confidentiality if the adopted erasure code meets the following condition: Any square submatrix of its generator matrix $G_{m\times n}$ is nonsingular.
\end{theorem}
\begin{proof}
We prove this theorem by contradiction as follows: In this case, suppose this IDA has weak confidentiality. According to Theorem~\ref{th-weak}, in $G_{m\times n}$, there is a submatrix $A_{m'\times n'}$ of column rank $r$, where $m',n'<m$ and $n'-r=m-m'>0$. Then, according to the proof of Theorem~\ref{th-weak}, $A_{m'\times n'}$ is rank deficient. Thus, in $A_{m'\times n'}$, any $\min(m',n')\times \min(m',n')$ square submatrix is singular---a contradiction! Therefore, this IDA has strong confidentiality.
\end{proof}

\indent In Rabin's IDA \cite{IDA}, the corresponding generator matrix is a Cauchy matrix, in which any square submatrix is nonsingular. Thus, Rabin's IDA has strong confidentiality.\\
\indent Inspired by the work in \cite{square-submatrix}, we now present an effective way to construct an IDA with strong confidentiality from an arbitrary \mbox{$m$-of-$(m+n)$} erasure code as follows:
\begin{enumerate}
  \item Choose an arbitrary \mbox{$m$-of-$(m+n)$} erasure code, whose generator matrix is $G_{m\times (m+n)}=\left(C_{m\times m}|D_{m\times n}\right)$;
  \item Construct an IDA that adopts an \mbox{$m$-of-$n$} erasure code whose generator matrix is $C_{m\times m}^{-1}\cdot D_{m\times n}$.
\end{enumerate}

\indent It is easy to verify that the IDA constructed above has strong confidentiality as follows:
\begin{enumerate}
  \item In the case where $C_{m\times m}$ is an $m\times m$ identity matrix, the chosen \mbox{$m$-of-$(m+n)$} erasure code is a systematic erasure code. Then, according to the nature of a  systematic \mbox{$m$-of-$(m+n)$} erasure code \cite{EEC}, any square submatrix of $D_{m\times n}=C_{m\times m}^{-1}\cdot D_{m\times n}$ is nonsingular. Thus, according to Theorem~\ref{th-strong}, the constructed IDA has strong confidentiality.
  \item In the case where $C_{m\times m}$ is not an $m\times m$ identity matrix, the chosen \mbox{$m$-of-$(m+n)$} erasure code is a \mbox{non-systematic} erasure code. Then, $\left(I_{m\times m}|C_{m\times m}^{-1}\cdot D_{m\times n}\right)$ is the generator matrix of the equivalent systematic \mbox{$m$-of-$(m+n)$} erasure code. So, any square submatrix of $C_{m\times m}^{-1}\cdot D_{m\times n}$ is nonsingular. Thus, according to Theorem~\ref{th-strong}, the constructed IDA also has strong confidentiality.
\end{enumerate}

\indent \emph{Example:} We construct an IDA with strong confidentiality from a Reed-Solomon code \cite{RS}, whose generator matrix is a Vandermonde matrix. From what we have discussed above, we first choose a \mbox{$m$-of-$(m+n)$} Reed-Solomon code with generator matrix
\begin{equation}
    G_{\mbox{RS}}=
    \left(
     \begin{array}{cccc}
         a_1^0 & a_2^0 & \cdots & a_{m+n}^0 \\
         a_1^1 & a_2^1 & \cdots & a_{m+n}^1 \\
         \vdots & \vdots & \ddots & \vdots \\
         a_1^{m-1} & a_2^{m-1} & \cdots & a_{m+n}^{m-1} \\
     \end{array}
\right),
\end{equation}
where $a_1,a_2,\cdots,a_{m+n}$ are distinct. Then, an IDA with strong confidentiality can be reconstructed, in which the corresponding generator matrix is
\begin{align}
G_{\mbox{IDA}}=&
\left(
     \begin{array}{cccc}
         a_1^0 & a_2^0 & \cdots & a_{m}^0 \\
         a_1^1 & a_2^1 & \cdots & a_{m}^1 \\
         \vdots & \vdots & \ddots & \vdots \\
         a_1^{m-1} & a_2^{m-1} & \cdots & a_{m}^{m-1} \\
     \end{array}
\right)^{-1}
\times \nonumber \\
&\left(
     \begin{array}{cccc}
         a_{m+1}^0 & a_{m+2}^0 & \cdots & a_{m+n}^0 \\
         a_{m+1}^1 & a_{m+2}^1 & \cdots & a_{m+n}^1 \\
         \vdots & \vdots & \ddots & \vdots \\
         a_{m+1}^{m-1} & a_{m+2}^{m-1} & \cdots & a_{m+n}^{m-1} \\
     \end{array}
\right).
\end{align}
From the comparison results  in \cite{Vandermonde}, we can deduce that the computation complexity of this IDA is comparable to or sometimes even lower than that of Rabin's IDA. \\
\indent \emph{Remark:} Besides Cauchy matrices, Vandermonde matrices were also suggested for the generator matrices of IDAs in Rabin's seminal paper \cite[Page 339]{IDA}. However, a Vandermonde matrix defined over a finite field may contain singular square submatrices \cite[Page 323, Problem. (7)]{EEC}. Then, an IDA whose erasure code is a Reed-Solomon code defined over a finite field may not meet the condition in Theorem~\ref{th-strong} and thus may have weak confidentiality. Luckily, in the literature, when Rabin's IDA is mentioned, it always refers to that constructed based on a Cauchy matrix.

\section{Conclusions}\label{section-conclusion}
\indent \emph{Information Dispersal Algorithms (IDAs)} \cite{IDA} have been widely applied to reliable and secure storage and transmission of data files in distributed systems. This paper made a systematic study on the \emph{confidentiality} of IDAs and its connection with the adopted erasure codes \cite{EEC}. Specially, this paper studied the confidentiality of IDAs from the view of practical security. This paper defined and discussed two levels of confidentiality: \emph{weak confidentiality} (in the case where some parts of the original file can be reconstructed explicitly from fewer than the threshold number of pieces) and \emph{strong confidentiality} (in the case where nothing of the original file can be reconstructed explicitly from fewer than the threshold number of pieces). This paper showed which kind of IDAs have weak confidentiality and how an eavesdropper can reconstruct some segments of the original file explicitly from fewer than the threshold number of pieces in the case of weak confidentiality (see Theorem~\ref{th-weak}). It was noticed that for an IDA that adopts an arbitrary \mbox{non-systematic} erasure code, its confidentiality may fall into weak confidentiality. To achieve strong confidentiality, this paper explored a sufficient and feasible condition for an IDA (see Theorem~\ref{th-strong}). It was showed that Rabin's IDA \cite{IDA} has strong confidentiality. At the same time, this paper presented an effective way to construct an IDA with strong confidentiality. Then, as an example, this paper constructed an IDA with strong confidentiality from a Reed-Solomon code \cite{RS}, the computation complexity of which is comparable to or sometimes even lower than that of Rabin's IDA.\\
\indent The key message we want to deliver through this paper is that \emph{an arbitrary \mbox{non-systematic} erasure code is not enough for strong confidentiality in an IDA}. When referring to the confidentiality of an IDA in a practical application, we should keep this point in mind!


\bibliographystyle{IEEEtran}
\bibliography{StrongIDA_Li}

\begin{IEEEbiography}[{\includegraphics[width=1in,height=1.25in,clip,keepaspectratio]{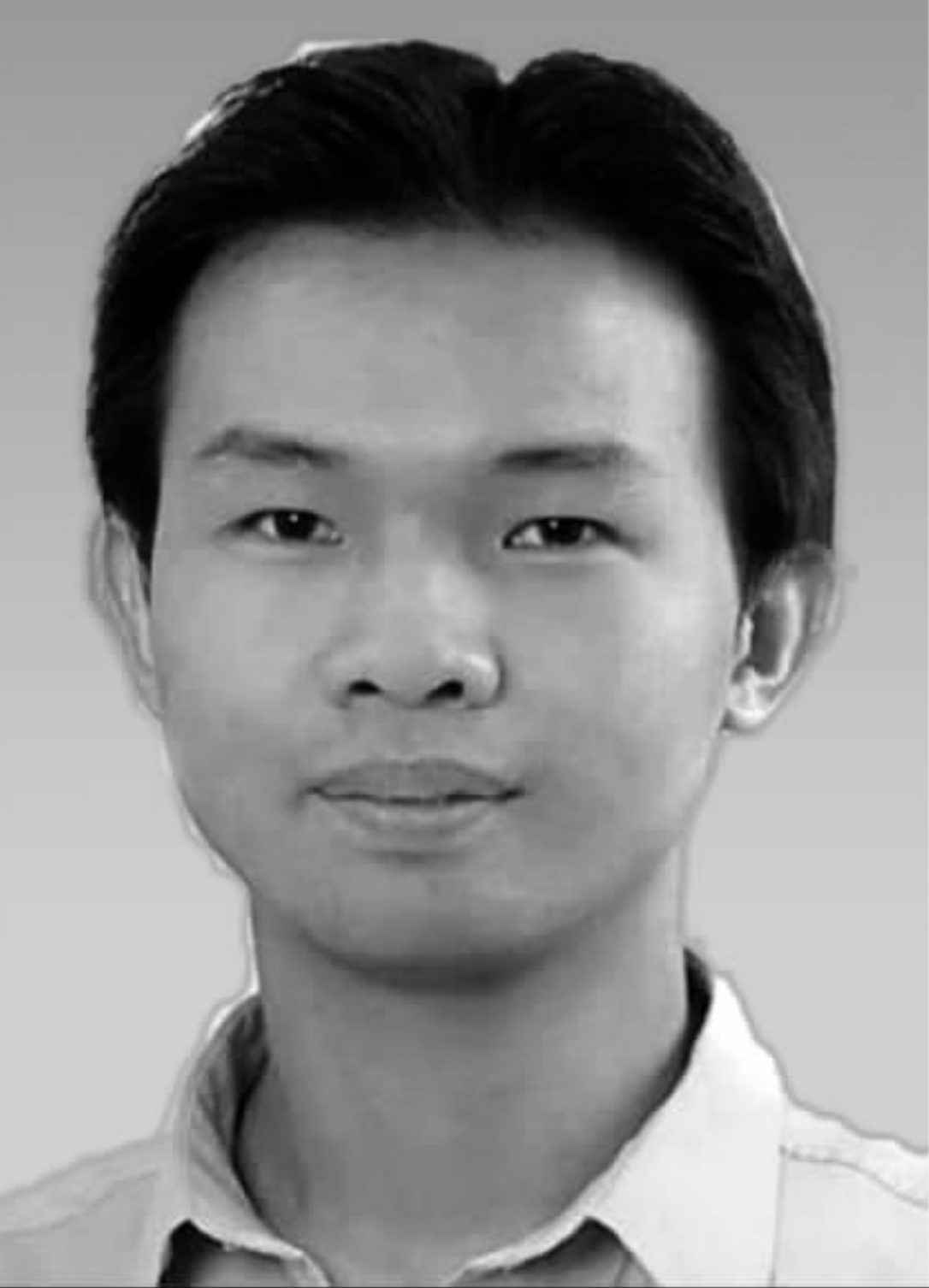}}]{Mingqiang Li}
received a Ph.D. degree (with honor) in Computer Science from Tsinghua University in July, 2011. He also received a B.S. degree in Mathematics from the University of Electronic Science and Technology of China in July, 2006. He worked as a Staff Researcher in the IBM China Research Laboratory from July, 2011 to February, 2013. He is now a Postdoctoral Fellow in the Department of Computer Science and Engineering, The Chinese University of Hong Kong. His current research interests include coding theory, storage systems, data security, data compression, cloud infrastructure, distributed systems, wireless networking, and network economics.
\end{IEEEbiography}

\end{document}